\documentclass[journal]{IEEEtran}

\usepackage[latin1]{inputenc}
\usepackage{graphicx}
\usepackage{amssymb,amsmath,amsfonts,amsthm}
\usepackage{algorithm,algorithmic}
\usepackage{cite}
\usepackage{float}
\usepackage{graphics}
\usepackage{subfigure}
\usepackage{xspace}
\usepackage[usenames,dvipsnames]{color}
\usepackage{epsfig, hyperref}
\usepackage{dsfont}

\newtheorem{theorem}{Theorem}
\newtheorem{proposition}{Proposition}
\newtheorem{remark}{Remark}

\begin{document}

\title{Average Age of Information in Wireless \\Powered Sensor Networks}

\author{Ioannis Krikidis,~\IEEEmembership{Fellow,~IEEE}
\thanks{I. Krikidis is with the Department of Electrical and Computer Engineering, Faculty of Engineering, University of Cyprus, Nicosia 1678 (E-mail: {\sf krikidis@ucy.ac.cy}).}}

\maketitle

\vspace{-0.4cm}
\begin{abstract}
In this letter, we deal with the age of information (AoI) for a sensor network with wireless power transfer (WPT) capabilities. Specifically,  we study a simple network topology, where a sensor node harvests energy from radio frequency signals (transmitted by a dedicated energy source) to transmit real-time status updates. The sensor node generates an update when its capacitor/battery becomes fully charged and transmits by using all the available energy without further energy management. The average AoI performance of the considered greedy policy is derived in closed form and is a function of the capacitor's size. The optimal value of the capacitor that maximizes the freshness of the information, corresponds to a simple optimization problem requiring a one-dimensional search. The derived theoretical results provide useful performance bounds for practical WPT networks.
\end{abstract}

\vspace{-0.3cm}
\begin{keywords}
Age of information, wireless power transfer, energy harvesting, sensor networks. 
\end{keywords}

\vspace{-0.3cm}
\section{Introduction}

\IEEEPARstart{W}{ireless} power transfer (WPT) via dedicated radio-frequency (RF) radiation is a promising technology for wireless communication systems, which are characterized by a massive number of low-power devices such as in the Internet-of-Things systems. It can support mobility, energy multicasting, non-line-of-sight propagation environments and contributes in the development of smaller, lighter and more compact devices. From the pioneering work of Varshney \cite{VAR}, who has introduced this concept, WPT has been extensively studied in the literature for different network architectures e.g., \cite{RUI2,RUI3}. However, most of the current works focus on complex network structures with limited practical interest and/or use conventional performance metrics e.g.,  throughput, coverage probability, diversity gain, information-energy capacity etc, which do not capture timeliness requirements that arise from sensing and actuation applications within machine-type communications. 

A performance metric that captures the freshness of the received information and is appropriate for applications requiring timely information to accomplish specific tasks (e.g., sensor networks, cyberphysical systems, etc), was proposed in \cite{ROY} i.e., {\it age of information (AoI)}. It is defined  as the time elapsed since the generation of the freshest status update that has reached the destination. Initial works on AoI take into account traffic burstiness and  minimize the AoI from a queueing-theoretic standpoint under various service policies e.g., \cite{EPH,KOS}. Recent works employ the notion of AoI in energy harvesting communication systems (from natural renewable sources), and investigate transmission policies that minimize AoI-based performance metrics \cite{ULU1,ULU2,WU,ULU3}. On the other hand, the design of WPT-based communication systems with objective  to optimize AoI, is a new research area with potential applications. In \cite{DON}, the authors propose a two-way data exchanging system, where a master node transfers energy and information to a slave node, while the slave node uses the energy harvested to power the uplink channel; the average uplink AoI is derived in closed form.  Although AoI seems to be a natural design metric for WPT networks, other relevant works cannot be found in the literature. 

\begin{figure}[t]
\centering
 \includegraphics[width=0.85\linewidth]{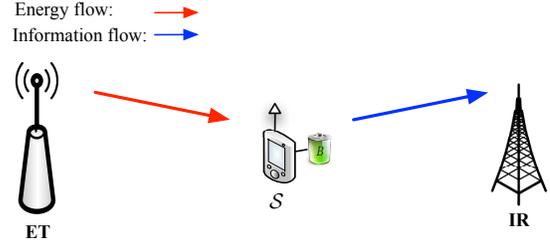}\\
\vspace{-0.5cm}
\caption{A three node sensor network topology; ET broadcasts energy, $\mathcal{S}$ communicates with IR by discharging its capacitor of size $B$.}\label{model}
\end{figure}

In this letter, we study a basic communication link where a sensor node with WPT capabilities communicates with a single destination. Specifically, the sensor node is equipped with a capacitor, which is charged via RF radiation by a dedicated energy source. Once the capacitor is charged, the sensor node transmits status updates containing the most recent information about parameters of interest by using all the energy stored. This online transmission policy does not require complicated energy management decisions (e.g., energy-depended thresholds) and is appropriate for WPT low complexity/low power devices. We investigate the freshness of the  received information and we provide simple closed form expressions for the average AoI, which depend on the size of the capacitor. The design of the system introduces an interesting tradeoff: a small capacitor is charged quickly and thus new updates are sent more frequently to minimize the AoI; on the other hand, a larger capacitor increases the transmit power and boosts the successful decoding. The optimal value of the capacitor is computed by solving a one-dimensional optimization problem. It is worth noting that the network topology and the transmission policy considered are inspired by commercial battery-free WPT products e.g., Powercast \cite{POW}; these devices are equipped with supercapacitors that deliver high power bursts when charged. Although our analysis refers to a simplistic system model, the derived theoretical results can serve as guidelines (performance bounds) for practical implementations.

\vspace{-0.4cm}
\section{System model}\label{stmod}

We assume a simple WPT sensor network consisting of one energy transmitter (ET), one sensor node, $\mathcal{S}$, and one information receiver (IR); all the nodes are equipped with single antennas. The ET is connected to the power grid and continuously broadcasts an energy signal with power $P$. The sensor node has WPT capabilities and harvests energy from the received RF signal; the harvested energy is stored in a capacitor of finite-size $B$. When the capacitor becomes fully charged, the sensor node generates a status update and transmits it towards the IR by using all the stored energy (greedy online policy \cite{POW}). Energy transmission and communication links are performed in orthogonal channels (e.g., different frequency bands) to avoid interference; in addition, time is considered to be slotted with a slot size equal to one time unit (due to the normalized slot duration, the measures of energy and power become identical and therefore are used interchangeably throughout the paper). Fig. \ref{model} schematically presents the system model.

The sensor node is able to harvest energy from the ET during the status transmission. This is feasible due to the orthogonality between the communication/harvesting links and the existence of an appropriate capacitor architecture that supports simultaneous transmission/harvesting (i.e., two antennas operating in different frequency bands for communication/harvesting;  a secondary storage capacitor/device stores up harvested energy while the transmitter is active \cite{RUI2}).

All wireless links experience Rayleigh block fading (channel is constant for one time slot and changes independently across time slots). Let $h_k,g_k\sim \exp(\lambda)$ denote the power of the channel fading for the link ET-$\mathcal{S}$ and $\mathcal{S}$-IR at the $k$-th time slot, respectively. In addition, all wireless links exhibit additive white Gaussian noise (AWGN) with variance $\sigma^2$. The energy stored (i.e., the amount of available energy in the capacitor) at time slot $k$, denoted as $E_k$, will evolve as follows\footnote{A linear WPT model is sufficient for the purposes of this work \cite{RUI3,DON} and provides useful lower bounds for the harvested energy achieved by non-linear models.} 
\begin{align}
E_{k}=\min\{\mathds{1}_{E_{k-1}<B}E_{k-1}+\eta P h_k, B\},
\end{align}
where $0 \leq \eta \leq 1$ denotes the RF-to-DC conversion efficiency (harvesting from the AWGN is considered negligible) and $\mathds{1}_X$ is the indicator function of $X$, with $\mathds{1}_X=1$ if $X$ is true and $\mathds{1}_X=0$ otherwise. If the capacitor becomes fully charged at time slot $k$ i.e., $E_{k}=B$, the sensor node transmits a status update to the IR, containing information for the considered parameters of interest as well as the time of generation of the update, with a spectral efficiency $r$ bits per channel use (BPCU) in the $(k+1)$-th time slot (a packet transmission is performed in one time slot). The signal-to-noise ratio at the IR for the $k$-th time slot is written as 
\begin{align}
\gamma_k=\frac{Bg_k}{\sigma^2}.
\end{align}

\noindent {\it Age of information:} In time slot $n$, AoI is the difference between $n$ and the generation time $U(n)$ measured in time slots of the latest received update at the IR \cite{ROY} i.e., 
\begin{align}
\Delta(n)=n-U(n).
\end{align}
Fig. \ref{model2} presents an example of the age evolution for the sensor network considered. An update is generated at the sensor node when the capacitor becomes fully charged and transmitted in the next slot (one time slot of delay); in case of a successful decoding i.e.,  $\log_2(1+\gamma_k)\geq r$, the AoI at the IR is reset to one. If $n_k$, $n_{k+1}$ represent the time slots of two consecutive updates at the IR,  $X_k=n_{k+1}-n_{k}$ denotes the $k$-th interarrival time  (time  between $n_{k+1}$ and $n_k$ in time slots). In addition, $T_k$ denotes the time (in time slots) between two consecutive capacitor recharges; we have $X_k=\sum_{i=1}^{M}T_i$, where $M$ is a discrete random variable that denotes the number of the update transmissions until successful decoding.

\begin{figure}
\centering
 \includegraphics[width=0.8\linewidth]{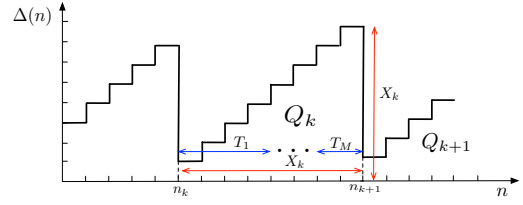}\\
\vspace{-0.3cm}
\caption{Example of AoI; $X_k$ denotes the interarrival time between two consecutive received updates, $T_k$ is the time between two consecutive capacitor's recharges, $Q_k$ is the area under $\Delta(n)$ corresponding to the $k$-th received update.}\label{model2}
\end{figure}

\section{Analysis of the average Age of Information}
In this section, we analyze the performance of the sensor network considered in terms of the average AoI. Firstly, we state two propositions which are used to derive the AoI performance.

\begin{proposition}\label{prop1}
The first-order and second-order moments of the time between two consecutive capacitor recharges, respectively, are given by
\begin{align}
&\mathsf{E}(T)=1+\beta, \\
&\mathsf{E}(T^2)=1+3\beta+\beta^2.
\end{align}
where $\beta=\lambda B/(\eta P)$.
\end{proposition}
\begin{proof}
See Appendix \ref{proof1}.
\end{proof}

Based on Proposition \ref{prop1}, we have the following proposition for the first-order and the second-order moments for the time between two consecutive successful delivered status updates. 

\begin{proposition}\label{prop2}
The first-order and the second-order moments for the interarrival time between two consecutive updates at the IR, respectively, are given by
\begin{align}
&\mathsf{E}(X)=\frac{1+\beta}{\pi}, \\
&\mathsf{E}(X^2)=\frac{1+3\beta+\beta^2}{\pi}+\frac{2(1+\beta)^2(1-\pi)}{\pi^2},
\end{align}
where $\pi=\mathsf{P}\{\log_2(1+\gamma_k)>r \}=\exp(-\lambda \frac{2^r-1}{B/\sigma^2})$ is the success probability for the link $\mathcal{S}$-IR (Rayleigh fading).
\end{proposition}
\begin{proof}
See Appendix \ref{proof2}.
\end{proof}

For a time period of $N$ time slots where $K$ successful transmissions occur, the average AoI can be written as 
\begin{align}
\Delta_N &=\frac{1}{N} \sum_{n=1}^{N}\Delta(n)=\frac{1}{N}\sum_{k=1}^{K}Q_k=\frac{K}{N}\frac{1}{K}\sum_{k=1}^{K}Q_k,
\end{align}
where $Q_k$ denotes the area under $\Delta(n)$ corresponding to the $k$-th status update.
The time average $\Delta_N$ tends to the ensemble average age for $N\to \infty$ \cite{KOS} i.e., 
\begin{align}
\Delta=\lim_{N\rightarrow \infty} \Delta_{N}= \frac{\mathsf{E}(Q)}{\mathsf{E}(X)}, \label{res4}
\end{align}
where $\lim_{N\rightarrow \infty} \frac{K(N)}{N}=\frac{1}{\mathsf{E}(X)}$ is the steady state rate of updates generation.

The area under $\Delta(n)$ for the $k$-th update corresponds to the sum of $X_k$ rectangles with one side equal to one and the other side equal to $m$, with $1 \leq m \leq X_k$. Therefore, $Q_k$ can be written as
\begin{align}
Q_k=\sum_{m=1}^{X_k}m=\frac{X_k(X_k+1)}{2}.
\end{align}
By taking the expectation operator, the average area under $\Delta(n)$ can be expressed as 
\begin{align}
\mathsf{E}(Q)=\frac{\mathsf{E}(X^2)+\mathsf{E}(X)}{2}. \label{res3}
\end{align}
By using Propositions \ref{prop1} and \ref{prop2}, and by substituting \eqref{res3} in \eqref{res4}, we have the following theorem on the average AoI. 
\begin{theorem}\label{th1}
The average AoI for the considered sensor network is given by
\begin{align}
\Delta&=\frac{1}{2}\left( \frac{\mathsf{E}(X^2)}{\mathsf{E}(X)}+1 \right)=\frac{1+3\beta+\beta^2}{2(1+\beta)}+\frac{(1+\beta)(1-\pi)}{\pi}+\frac{1}{2}. \label{n1}
\end{align}
\end{theorem}
From Theorem \ref{th1}, we have the following two remarks. 
\begin{remark}\label{rm1}
For the case where $P\rightarrow \infty$ and $B$ is a constant, the average AoI asymptotically converges to $\Delta\rightarrow 1/\pi$.
\end{remark}

\begin{remark}
For the case where $P\rightarrow \infty$ and $B\rightarrow \infty$ with a ratio $B/P=\theta$, the average AoI asymptotically converges to $\Delta\rightarrow  \frac{1+3(\lambda \theta/\eta)+(\lambda \theta/\eta)^2}{2(1+\lambda \theta/\eta)}+\frac{1}{2}$.
\end{remark}

If the objective of the system is to design the capacitor $B$ such as the IR has as much as possible fresh information, we introduce the following one-dimensional optimization problem; we assume that the transmit power is given and we minimize the AoI with respect to the size of the capacitor $B$. The optimal capacitor size is given by 
\begin{align}
B^*=\arg \min_{B>0} \Delta.  \label{op1}
\end{align} 

Given \eqref{n1}, unfortunately the optimization problem in \eqref{op1} does not admit closed-form solutions; however, the optimal $B^*$ can be solved numerically (e.g., {\it fminsearch} in MATLAB).

\section{Numerical results}

The simulation setup follows the description of Section \ref{stmod} with parameters $\sigma^2=-50$ dBm, $\eta=0.5$, $r=0.05$ BPCU; the sensor node is located $20$ meters away from both the ET and the IR; the channel power gains are modeled as $\lambda=10^{3} d^{\alpha}$, where $d$ is the link distance and  $\alpha=2.2$ is the path-loss exponent \cite{QIN}.

\begin{figure}
\centering
 \includegraphics[width=0.8\linewidth]{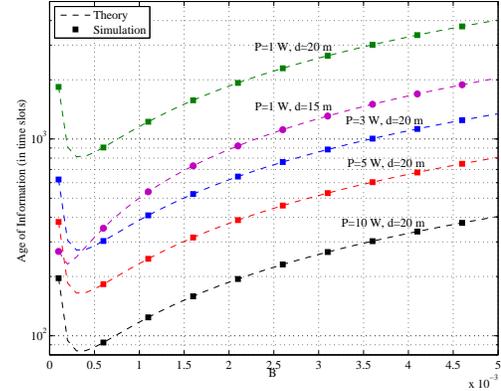}\\
\vspace{-0.3cm}
\caption{Average AoI versus capacitor's size $B$ for $P=\{1, 3, 5, 10\} $ Watt.}\label{fg1}
\end{figure}

\begin{figure}
\centering
 \includegraphics[width=0.8\linewidth]{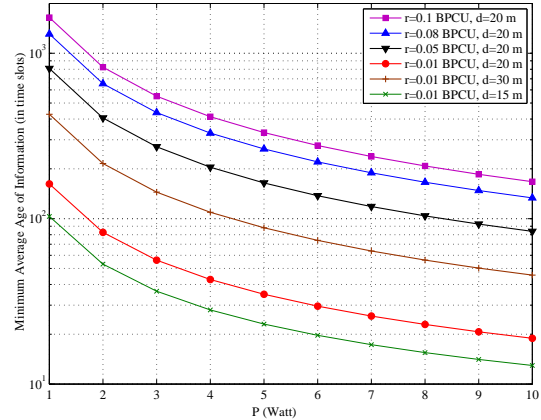}\\
\vspace{-0.3cm}
\caption{Minimun Average AoI versus $P$ for $r=\{0.01, 0.05, 0.08, 0.1 \}$ BPCU.} \label{fg2}
\end{figure}

Fig. \ref{fg1} depicts the average AoI versus the capacitor's size $B$ for different values of $P$ and $d$. As can be seen, the parameter $B$ significantly affects the average AoI performance of the system. Specifically, a high $B$ facilitates the transmission phase but requires more time slots to charge the capacitor, while a low $B$ activates the sensor node faster but the available energy for transmission is low.  The optimal value of $B$ is different for each $P$ value and matches with the optimal solutions given by the optimization problem in \eqref{op1}. We can also observe, that as $P$ (and/or $d$) increases the achieved average AoI decreases; a higher $P$ (and/or shorter $d$) charges the capacitor faster and therefore decreases the time that the sensor is idle (in the energy harvesting mode).  Finally, theoretical results perfectly match with the simulation results and validate the analysis. 

Fig. \ref{fg2} plots the minimum average AoI (corresponding to $B^*$) versus $P$ for different spectral efficiencies and distances. As it is expected, for a given $P$, a higher spectral efficiency (and/or distance) increases the achieved AoI, as it requires more transmission attempts before a successful transmission takes place; asymptotically ($P\rightarrow\infty$), the average AoI converges to an AoI floor that is equal to $1/\pi$ (see Remark \ref{rm1}).     

\vspace{-0.3cm}
\section{Conclusion}

In this letter, we have studied the performance of a basic wireless powered sensor network in terms of the average AoI. The sensor node transmits updates to the destination by discharging its capacitor, which is charged by a dedicated energy source. We derived simple closed form expressions for the average AoI and we showed that it highly depends on the capacitor's size. The optimal capacitor value has been computed numerically by formulating and solving a one-dimensional optimization problem. 

\appendices

\vspace{-0.3cm}
\section{Probability of capacitor charging in $K$ consecutive time slots} \label{app1}

Let $X_1,\ldots,X_K$ denote $K$ independent and identically distributed exponential random variables with rate parameter $\lambda$. We calculate the following probability
{\small
\begin{align}
&\Pi(y,K)=\mathsf{P} \left\{ \left( \sum_{i=1}^{K-1}X_i<y \right)\bigcap \left( \sum_{i=1}^{K}X_i\geq y \right) \right\} \nonumber \\
&=\int_{0}^{y}\!\!\int_{0}^{y-x_1}\!\!\!\!\!\!\!\!\!\!\!\ldots \int_{0}^{y-\sum_{i=1}^{K-2}x_i}\!\!\!\!\int_{y-\sum_{i=1}^{K-1}x_i}^{\infty}\!f_{X_1,\ldots,X_K}(x_1,\dots,x_K) dx_K\ldots dx_1 \nonumber \\
&=\int_{0}^{y}\!\!\int_{0}^{y-x_1}\!\!\!\!\!\!\!\!\!\!\!\ldots \int_{0}^{y-\sum_{i=1}^{K-2}x_i}\!\!\!\!\int_{y-\sum_{i=1}^{K-1}x_i}^{\infty}\!\lambda^K \exp\left(-\lambda \sum_{i=1}^{K}x_i \right)dx_K\ldots dx_1 \nonumber \\
&=\frac{1}{(K-1)!}(\lambda y)^{K-1}\exp(-\lambda y),
\end{align}} 
where $f_{X_1,\ldots,X_K}(x_1,\ldots, x_K)=\prod_{i=1}^{K} f_{X}(x_i)$ is the joint probability density function (PDF) of the random variables $X_1,\ldots,X_K$, and $f_X(x)=\lambda \exp(-\lambda x)$ is the PDF for an exponential random variable with parameter $\lambda$. 

By using the above computation, the probability that the capacitor is charged in $K$ consecutive time slots becomes equal to $\Pi(\beta/\lambda,K)$.

\vspace{-0.3cm}
\section{Proof of Proposition 1}\label{proof1}

The average time between two consecutive capacitor's recharges can be  computed as follows
\begin{align}
\mathsf{E}(T)&=\sum_{k=1}^{\infty}k\mathsf{P}\{T=k\}=\sum_{k=1}^{\infty}k\Pi(\beta/\lambda,k) \nonumber \\
&=\exp(-\beta)\sum_{k=1}^{\infty}\frac{k}{(k-1)!}\beta^{k-1}=\beta+1, \label{res1}
\end{align}
where $\Pi(\beta/\lambda,k)$ is given in Appendix \ref{app1}, and the result in \eqref{res1} is based on \cite[Eq. 1.212]{GRAD}. For the second-order moment of the time between two consecutive recharges, we have
\begin{align}
\mathsf{E}(T^2)&=\sum_{k=1}^{\infty}k^2\mathsf{P}\{T=k\}=\sum_{k=1}^{\infty}k^2\Pi(\beta/\lambda,k)  \nonumber \\
&=\exp(-\beta) \sum_{k=1}^{\infty}\frac{k^2}{(k-1)!}\beta^{k-1}=1+3\beta+\beta^2. \label{t2}
\end{align}

\vspace{-0.3cm}
\section{Proof of Proposition 2}\label{proof2}

The interarrival time can be written as $X=\sum_{i=1}^k T_i$, where $k$ denotes the number of the consecutive transmissions until successful decoding at the IR and it is a (positive integer) random variable. If $k$ transmissions occur, this means that $(k-1)$ consecutive transmissions were unsuccessful, while the $k$-th transmission was successful. Therefore, the average interarrival time becomes equal to
\begin{align}
\mathsf{E}(X)=\sum_{k=1}^{\infty} k\mathsf{E}(T)(1-\pi)^{k-1} \pi\;=\frac{\beta+1}{\pi}, \label{res2}
\end{align}
where $\pi$ denotes the success probability for the link $\mathcal{S}$-IR, and \eqref{res2} is based on \cite[Eq. 1.113]{GRAD}. 

For the second-order moment, we have
\begin{align}
X^2=\left(\sum_{i=1}^k T_i \right)^2=\sum_{i=1}^k T_i^2+2 \sum_{i=1}^{k}\sum_{j>i}^{k} T_i T_j.
\end{align} 
By taking the conditional expectation operator and after some basic manipulations, we have
\begin{align}
\mathsf{E}(X^2|k)=k \mathsf{E} (T^2)+k(k-1) \mathsf{E}(T)^2.
\end{align}
By using similar arguments with the computation of the first-order moment, we average out the number of transmissions i.e., 
\begin{align}
\mathsf{E}(X^2)&= \sum_{k=1}^{\infty} \mathsf{E}(X^2|k) (1-\pi)^{k-1}\pi \nonumber \\
&=\frac{E(T^2)}{\pi}+\mathsf{E}(T)^2\frac{2(1-\pi)}{\pi^2} \nonumber \\
&=\frac{1+3\beta+\beta^2}{\pi}+\frac{2(1+\beta)^2(1-\pi)}{\pi^2}, \label{final}
\end{align}
where \eqref{final} is based on the expressions in \eqref{res1}, \eqref{t2}.

\end{document}